\documentclass[11pt]{article}

\usepackage{amssymb,amsmath,amsthm,amsfonts}
\usepackage[margin=1in]{geometry}
\usepackage{colonequals}
\usepackage[colorlinks]{hyperref}
\usepackage[all]{hypcap}
\usepackage{tikz}
\usetikzlibrary{calc}
\usepackage{url}
\newtheorem{theorem}{Theorem}
\newtheorem{lemma}{Lemma}

\newcommand{\fig}[1]{\hyperref[fig:#1]{Figure~\ref*{fig:#1}}}
\renewcommand{\sec}[1]{\hyperref[sec:#1]{Section~\ref*{sec:#1}}}
\newcommand{\thm}[1]{\hyperref[thm:#1]{Theorem~\ref*{thm:#1}}}
\newcommand{\lem}[1]{\hyperref[lem:#1]{Lemma~\ref*{lem:#1}}}

\newcommand{\floor}[1]{\lfloor{#1}\rfloor}
\newcommand{\ceil}[1]{\lceil{#1}\rceil}

\newcommand{\N}{\mathbb{N}}
\newcommand{\Z}{\mathbb{Z}}

\newcommand{\defeq}{\colonequals}

\DeclareMathOperator{\enc}{enc}
\DeclareMathOperator{\poly}{poly}

\renewcommand{\>}{\rangle}
\newcommand{\<}{\langle}

\newcommand{\upper}{N}
\newcommand{\sample}{M}
\newcommand{\numpoints}{L}
\newcommand{\neutral}{0}

\begin{document}

\title{Quantum computation of discrete logarithms in semigroups}

\author{Andrew M.\ Childs \\
\normalsize{Department of Combinatorics \& Optimization} \\
\normalsize{and Institute for Quantum Computing} \\
\normalsize{University of Waterloo}
\and
G{\'a}bor Ivanyos \\
\normalsize{Institute for Computer Science and Control} \\
\normalsize{Hungarian Academy of Sciences}}

\date{}

\maketitle

\begin{abstract}
We describe an efficient quantum algorithm for computing discrete logarithms in semigroups using Shor's algorithms for period finding and discrete log as subroutines.  Thus proposed cryptosystems based on the presumed hardness of discrete logarithms in semigroups are insecure against quantum attacks.
In contrast, we show that some generalizations of the discrete log problem are hard in semigroups despite being easy in groups.  We relate a shifted version of the discrete log problem in semigroups to the dihedral hidden subgroup problem, and we show that the constructive membership problem with respect to $k \ge 2$ generators in a black-box abelian semigroup of order $N$ requires $\tilde \Theta(N^{\frac{1}{2}-\frac{1}{2k}})$ quantum queries.
\end{abstract}

\section{Introduction}

The presumed difficulty of computing discrete logarithms in groups is a common cryptographic assumption. For example, such an assumption underlies Diffie-Hellman key exchange, ElGamal encryption, and most elliptic curve cryptography. While such cryptosystems may be secure against classical computers, Shor showed that quantum computers can efficiently compute discrete logarithms \cite{Sho97}. Shor originally described an algorithm for computing discrete logs in multiplicative groups of integers, but it is well known that his approach efficiently computes discrete logs in any finite group, provided only that group elements have a unique encoding and that group operations can be performed efficiently.

Here we consider the closely related problem of computing discrete logarithms in finite semigroups. A semigroup is simply a set equipped with an associative binary operation. In particular, a semigroup need not have inverses (and also need not have an identity element). We suppose that the semigroup elements have a unique encoding and that we are able to perform semigroup operations efficiently. In the discrete log problem for a semigroup $S$, we are given two elements $x,g \in S$ and are asked to find the smallest $a \in \N \defeq \{1,2,\ldots\}$ such that $g^a = x$ (or to determine that no such $a$ exists). We write $a = \log_g x$.

At first glance, it may be unclear how a quantum computer could compute discrete logs in semigroups. Shor's discrete log algorithm relies crucially on the function $(a,b) \mapsto g^a x^{-b}$, but $x^{-b}$ is not defined in a semigroup. In fact, hardness of the semigroup discrete logarithm problem has been proposed as a cryptographic assumption that might be secure against quantum computers \cite{KKS13}. The particular scheme described in \cite{KKS13}, based on matrix semigroups, has been broken by a quantum attack \cite{MU12}. However, the algorithm of \cite{MU12} uses a reduction from discrete logs in matrix groups to to discrete logs in finite fields \cite{MW97}, so it does not apply to general semigroups.

Here we point out that in fact quantum computers can efficiently compute discrete logs in any finite semigroup. Our approach is a straightforward application of known quantum tools. The structure of the semigroup generated by $g$ can be efficiently determined using the ability of a quantum computer to detect periodicity, as shown in \sec{structure}. Once this structure is known, an algorithm to compute discrete logs follows easily, as explained in \sec{dlog}.

On the other hand, some problems for semigroups are considerably harder than for groups. In \sec{sdlog}, we consider a shifted version of the discrete log problem in semigroups, namely solving the equation $x=yg^a$ for $a$. This problem appears comparably difficult to the dihedral hidden subgroup problem, even though the corresponding problem in a group can be solved efficiently by computing a discrete log. In \sec{hard}, we consider the problem of writing a given semigroup element as a product of $k \ge 2$ given generators of a black-box abelian semigroup. This problem can also be solved efficiently in groups, whereas the semigroup version is provably hard, requiring $\Omega(N^{\frac{1}{2}-\frac{1}{2k}})$ quantum queries. In fact, this bound is optimal up to logarithmic factors, as we show using the algorithm for the shifted discrete log problem.

After posting a preprint of this work, we learned of independent related work by Banin and Tsaban, who showed that the semigroup discrete log problem can be solved efficiently using an oracle for the discrete log problem in a cyclic group \cite{BT13}.  In particular, this implies a fast quantum algorithm for semigroup discrete log.

\section{Finding the period and index of a semigroup element}
\label{sec:structure}

Given a finite semigroup $S$, fix some element $g \in S$.  The value
\[
  t \defeq \min\{j \in \N \colon g^j = g^k \text{ for some $k \in j+\N$}\}
\]
is called the \emph{index} of $g$.  The index exists since $S$ is finite.  The value
\[
  r \defeq \min\{j \in \N \colon g^t = g^{t+j}\}
\]
is called the \emph{period} of $g$.  These definitions are illustrated in \fig{rho}.  If $j \ge t$, we say that $g^j$ is in the \emph{cycle} of $g$; if $j<t$, we say that $g^j$ is in the \emph{tail} of $g$.

\begin{figure}
\begin{center}
\begin{tikzpicture}
\tikzstyle{every node}=[font=\small];
\tikzstyle{arrow}=[->, >=latex, shorten >=3pt, shorten <=3pt];
\tikzstyle{noarrow}=[shorten >=3pt, shorten <=3pt];

\coordinate (g1) at (0,0);
\coordinate (g2) at (2,0);
\coordinate (g3) at (4,0);
\coordinate (gtm1) at (6,0);
\coordinate (gt) at (8,0);
\coordinate (center) at ($ (gt) - (0,5/pi) $);
\coordinate (gtp1) at ($ (center) + (90-360/5:5/pi) $);
\coordinate (gtp2) at ($ (center) + (90-2*360/5:5/pi) $);
\coordinate (gtprm2) at ($ (center) + (90-3*360/5:5/pi) $);
\coordinate (gtprm1) at ($ (center) + (90-4*360/5:5/pi) $);

\filldraw (g1) circle (1.5pt) node [above] {$g$};
\filldraw (g2) circle (1.5pt) node [above] {$g^2$};
\filldraw (g3) circle (1.5pt) node [above] {$g^3$};
\filldraw (gtm1) circle (1.5pt) node [above] {$g^{t-1}$};
\filldraw (gt) circle (1.5pt) node [above] {$g^t = g^{t+r}$};
\filldraw (gtp1) circle (1.5pt) node [right] {$g^{t+1}$};
\filldraw (gtp2) circle (1.5pt) node [below right] {$g^{t+2}$};
\filldraw (gtprm2) circle (1.5pt) node [below left] {$g^{t+r-2}$};
\filldraw (gtprm1) circle (1.5pt) node [left] {$g^{t+r-1}$};

\draw[arrow] (g1) -- (g2);
\draw[arrow] (g2) -- (g3);
\draw[noarrow] (g3) -- (4.75,0);
\draw[arrow] (5.25,0) -- (gtm1);
\draw[arrow] (gtm1) -- (gt);

\foreach \x in {-1,0,1} \filldraw (5,0)+(\x*.25,0) circle (.5pt);

\draw[arrow] (gt) arc (90:90-360/5:5/pi);
\draw[arrow] (gtp1) arc (90-360/5:90-2*360/5:5/pi);
\draw[noarrow] (gtp2) arc (90-2*360/5:-90+9:5/pi);
\draw[arrow] ($ (center) + (-90-9:5/pi) $) arc (-90-9:90-3*360/5:5/pi);
\draw[arrow] (gtprm2) arc (90-3*360/5:90-4*360/5:5/pi);
\draw[arrow] (gtprm1) arc (90-4*360/5:90-5*360/5:5/pi);

\foreach \x in {-1,0,1} \filldraw (center)+(\x*9-90:5/pi) circle (.5pt);
\end{tikzpicture}
\end{center}
\caption{\label{fig:rho}
The semigroup generated by $g$.}
\end{figure}
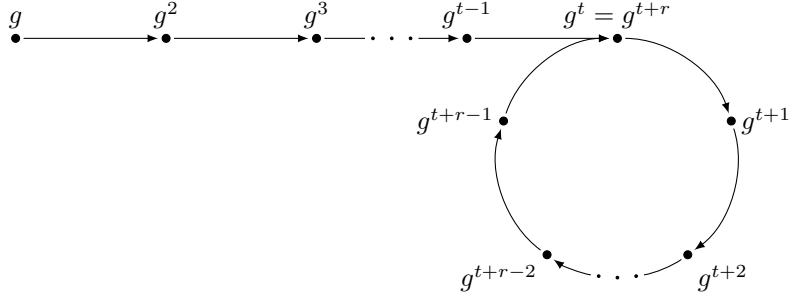

Suppose we are given an upper bound $t+r \le \upper$ on the order of $g$.  For a trivial upper bound, we could take $\upper$ to be the order of the semigroup.  We consider an algorithm to be efficient if it runs in time $\poly(\log\upper)$.  We claim that there is an efficient algorithm to compute $t$ and $r$.

\begin{lemma}\label{lem:indexperiod}
Given a black-box semigroup $S$ and an element $g \in S$, there is an efficient quantum algorithm to determine the index and the period of $g$.
\end{lemma}

\begin{proof}
First we find the period, as follows.  Create the state $\frac{1}{\sqrt \sample} \sum_{j=1}^\sample |j\>|g^j\>$ for some sufficiently large $\sample$ (it suffices to take $\sample > \upper^2+\upper$).  Note that we can compute $g^j$ efficiently even for exponentially large $j$ using repeated squaring, so this state can be made in polynomial time.  Next, we discard the second register.  To understand what happens when we do this, suppose we measure the second register.  If we obtain an element in the tail of $g$, then the first register is left in a computational basis state, which is useless.  However, with probability at least $(\sample-t+1)/\sample \ge 1 - \upper/\sample$, we obtain an element in the cycle of $g$, and we are left with an $r$-periodic state
\[
  \frac{1}{\sqrt{\numpoints}} \sum_{j=0}^{\numpoints-1} |x_0+jr\>
\]
for some unknown $x_0 \in \{t,t+1,\ldots,t+r-1\}$, where $\numpoints$ is either $\floor{(\sample-t)/r}$ or $\ceil{(\sample-t)/r}$ (depending on the value of $x_0$).  This is precisely the type of state that appears in Shor's period-finding algorithm. After Fourier transforming this state over $\Z_\sample$ and measuring, we obtain the outcome $k \in \Z_\sample$ with probability
\[
  \Pr(k) 
  = \frac{\sin^2(\frac{\pi k r \numpoints}{\sample})}
         {\numpoints \sample \sin^2(\frac{\pi k r}{\sample})}.
\]
A simple calculation (see for example \cite[Algorithm 5]{CD10}) shows that the probability of obtaining a closest integer to one of the $r$ integer multiples of $\sample/r$ is at least $4/\pi^2$.  By efficient classical postprocessing using continued fractions, we can recover $r$ from such values with constant probability \cite{Sho97}. Since we are in the cycle of $g$ with overwhelming probability, the overall procedure succeeds with constant probability (which could be boosted by standard techniques).

Given the period of $g$, we can find its index by an efficient classical procedure.  Observe that we can efficiently decide whether a given element $g^j$ is in the tail or the cycle of $g$: simply compute $g^r$ by repeated squaring and multiply by $g^j$ to compute $g^{j+r}$.  If $g^{j+r} = g^j$, then $g^j$ is in the cycle of $g$; otherwise it is in the tail of $g$.  Let
\[
  \gamma(g^j) \defeq \begin{cases}
  1 & \text{if $g^{j+r}=g^j$ (i.e., $g^j$ is in the cycle of $g$)} \\
  0 & \text{otherwise (i.e., $g^j$ is in the tail of $g$)}.
  \end{cases}
\]
The value of $t$ is precisely the index of the first $1$ in the list $(\gamma(g),\gamma(g^2),\ldots,\gamma(g^\upper))$.  This list consists of $t-1$ zeros followed by $\upper-t+1$ ones, so we can find $t$ in $O(\log\upper)$ iterations by binary search.
\end{proof}

\section{Computing discrete logs}
\label{sec:dlog}

We now show how to efficiently compute discrete logarithms in semigroups on a quantum computer.

\begin{theorem}\label{thm:dlog}
Given $x,g \in S$, there is an efficient quantum algorithm to compute $\log_g x$.
\end{theorem}

\begin{proof}
First, we use \lem{indexperiod} to compute the index $t$ and the period $r$ of $g$.  Then we determine whether $x$ is in the tail or the cycle of $g$.  As described in the proof of \lem{indexperiod}, this can be done efficiently by determining whether $xg^r=x$.

If $x$ is in the tail of $g$, then we compute $p$, the smallest positive integer such that $\gamma(xg^p)=1$.  This can be done efficiently by using binary search to find the first $1$ in the list $(\gamma(xg),\gamma(xg^2),\ldots,\gamma(xg^m))$.  Then we can compute $\log_g x = t - p$.

On the other hand, suppose $x$ is in the cycle of $g$.  Then we use the well-known fact (see for example \cite{How95}) that $C \defeq \{g^{t+j}\colon j \in \Z_r\}$ is a group with identity element $g^{t+s}$ where $s = -t \bmod r$.  In fact $C$ is a cyclic group generated by $g^{t+s+1}$; in particular, for $j \ge t$ we have $g^{t+s+1} g^j = g^{j+1}$.  Now we use Shor's discrete log algorithm to compute $\log_{g^{t+s+1}} x$.  While we cannot easily compute the inverse of $x$ in $C$, we know that the inverse of $g^{t+s+1}$ is $g^{t+s+r-1}$, so we can compute the hiding function $f\colon \Z_r \times \Z_r \to C$ with $f(a,b)=x^a g^{(t+s+r-1)b}=x^a (g^{t+s+1})^{-b}$, which suffices to efficiently compute discrete logs in $C$.  Thus we can compute $\log_g x = t + [(s + \log_{g^{t+s+1}} x) \bmod r]$.
\end{proof}

\section{A shifted version of discrete log}
\label{sec:sdlog}

While the discrete log problem is no harder in semigroups than in groups, some problems that have efficient quantum algorithms in groups are more difficult in semigroups.  In this section, we discuss a shifted version of the discrete log problem that appears to be closely related to the dihedral hidden subgroup problem.

The shifted discrete log problem is as follows: given $x,y,g \in S$, find some $a \in \N$ such that $x=y g^a$.  If $S$ is a group, then this problem reduces to the ordinary discrete log problem, since it suffices to find $a \in \N$ such that $g^a = y^{-1} x$.  However, if $S$ is a semigroup, then the best quantum algorithm we are aware of is the following.

\begin{lemma}\label{lem:sdlog}
Given a black-box semigroup $S$ and elements $x,y,g \in S$, there is a quantum algorithm to find $a \in \N$ such that $x=y g^a$ in time $2^{O(\sqrt{\log|S|})}$.  Furthermore, there is an algorithm using only $\poly(\log |S|)$ quantum queries.
\end{lemma}

\begin{proof}
Similarly to $j\mapsto g^j$, the function $j\mapsto yg^j$ has index
\[
  \tilde t \defeq
  \min\{j \in \N \colon yg^j = yg^k \text{ for some $k \in j+\N$}\}
\]
and period
\[
  \tilde r \defeq
  \min\{j \in \N \colon y g^{\tilde t} = g^{\tilde t + j}\};
\]
we say that $yg^j$ is in the cycle if $j \ge \tilde t$ and in the tail if $j < \tilde t$. The period $\tilde r$ and the index $\tilde t$ can be computed efficiently along the same lines as described in \sec{structure}.

The case where $x$ is in the tail can be treated as in \sec{dlog}. If $x$ is in the cycle, so that $x=yg^{\tilde t+\ell}$ for some nonnegative integer $\ell$, then we must solve a constructive orbit membership problem for a permutation action of the group $\Z_{\tilde r}$ on the set of elements of the form $yg^{{\tilde t}+j}$. Specifically, the action of $j' \in \Z_{\tilde r}$ is multiplication by $g^{j'}$ and we must find the element $\ell\in \Z_{\tilde r}$ transporting $yg^{{\tilde t}}$ to $x$. To this end we consider the efficiently computable function $f\colon \Z_2 \ltimes \Z_{\tilde r} \to S$ with $f(0,j)=yg^{\tilde t+j}$ and $f(1,j)=xg^j$. The function $f(0,j)$ is injective since it has period $\tilde r$. Furthermore, $f(1,j)=xg^j=yg^{\tilde t+\ell+j}=f(0,j+\ell)$, i.e., $f(1,j)$ is a shift of $f(0,j)$ by $\ell$. Therefore, $f$ hides the subgroup $\<(1,\ell)\>$ of the dihedral group $\Z_2 \ltimes \Z_{\tilde r}$ (i.e., it is constant on the cosets of this subgroup and distinct on different cosets). It follows that the Kuperberg sieve~\cite{Kup05} finds $\ell$ (and hence $a=\tilde t+\ell$) in time $2^{O(\sqrt{\log r})}$. Finally, since the dihedral hidden subgroup problem can be solved with only polynomially many quantum queries to the hiding function \cite{EH00}, we can solve the shifted discrete log problem in a black-box semigroup $S$ with only $\poly(\log |S|)$ queries.
\end{proof}

The dihedral hidden subgroup problem (DHSP) is apparently hard.  Despite considerable effort (motivated by a close connection to lattice problems \cite{Reg04}), Kuperberg's algorithm remains the best known approach, and it is plausible that there might be no efficient quantum algorithm. Note that the DHSP can be reduced to a quantum generalization of the constructive orbit membership problem, namely, orbit membership for a permutation action on pairwise orthogonal quantum states \cite[Proposition 2.2]{FIMSS03}. Thus, intuitively, a solution of the shift problem for a (classical) permutation action (such as in the shifted discrete log problem) should exploit that the action is on classical states, unless it also solves the DHSP.

In \sec{hard}, we describe another variant of the discrete log problem that is even harder than the shifted discrete log problem, requiring exponentially many queries.  We also show that our lower bound for that problem is nearly optimal using the algorithm of \lem{sdlog} as a subroutine.

\section{Constructive semigroup membership}
\label{sec:hard}

Given an abelian semigroup with generators $g_1,\ldots,g_k$ and an element $x \in \<g_1,\ldots,g_k\>$, the \emph{constructive membership problem} asks us to find $a_1,\ldots,a_k \in \N_0 \defeq \{0,1,2,\ldots\}$ with $a_1+\cdots+a_k \ge 1$ such that $x = g_1^{a_1} \cdots g_k^{a_k}$.  The notation $g_i^0$ simply indicates that no factor of $g_i$ is present, so solutions with $a_i=0$ for some values of $i$ are well defined even though the semigroup need not have an identity element.

This natural generalization of the discrete log problem is easy for abelian groups (see for example \cite[Theorem 5]{IMS03}). In that case, let $r_i \defeq |\<g_i\>|$ for all $i \in \{1,\ldots,k\}$, $r \defeq |\<g\>|$, and $L \defeq \Z_{r_1} \times \cdots \times \Z_{r_k} \times \Z_r$. The values $(r_1,\ldots,r_k,r)$ can be computed efficiently by Shor's order-finding algorithm \cite{Sho97}. Now consider the function $f\colon L \to G$ defined by $f(a_1,\ldots,a_k,b) = g_1^{a_1} \cdots g_k^{a_k} x^{-b}$. This function hides the subgroup $H \defeq \{(x_1,\ldots,x_k,x) \in L \colon g_1^{x_1} \cdots g_k^{x_k} = g^x\} \le L$, so generators of $H$ can be found in polynomial time \cite{ME99}. To solve the constructive membership problem, it suffices to find the solutions with $x=1 \bmod r$. This corresponds to a system of linear Diophantine equations, so it can be solved classically in polynomial time (see for example \cite[Corollary 5.3b]{Sch86}).

Here we show that the constructive membership problem in semigroups is considerably harder.  Specifically, given a black-box semigroup $S$, we need exponentially many quantum queries (in $\log |S|$) to solve the constructive membership problem with respect to $k \ge 2$ generators.

\begin{theorem}\label{thm:lower}
Given a black-box semigroup $S = \<g_1,\ldots,g_k\>$ and an element $x \in S$, at least $\Omega(|S|^{\frac{1}{2}-\frac{1}{2k}})$ quantum queries are required to solve the constructive membership problem for $x$ with respect to $g_1,\ldots,g_k$.
\end{theorem}

\begin{proof}
For any $n \in \N$, consider the abelian semigroup
\[
  S = \{g_1^{a_1} \cdots g_k^{a_k}\colon a_1,\ldots,a_k \in \N_0,\, 1 \le a_1+\cdots+a_k \le n\} \cup \{\neutral\}
\]
with the following multiplication rules:
\begin{align*}
  \neutral (g_1^{a_1} \cdots g_k^{a_k}) &= \neutral \\
  (g_1^{a_1} \cdots g_k^{a_k}) (g_1^{b_1} \cdots g_k^{b_k}) 
  &= \begin{cases}
    g_1^{a_1+b_1} \cdots g_k^{a_k+b_k} &
    \text{if $a_1 + \cdots + a_k + b_1 + \cdots + b_k \le n$} \\
    \neutral & \text{if $a_1 + \cdots + a_k + b_1 + \cdots + b_k > n$}.
  \end{cases}
\end{align*}
Let $\Sigma \defeq \{(a_1,\ldots,a_{k-1}) \in \N_0^{k-1} \colon a_1+\cdots+a_{k-1} \le n\}$. We show that the problem of inverting a black-box permutation $\pi\colon \Sigma \to \Sigma$ (i.e., computing $\pi^{-1}(\sigma)$ for any fixed $\sigma \in \Sigma$ given a black box for $\pi$) reduces to constructive semigroup membership in a black-box version of $S$ with respect to the generators $g_1,\ldots,g_k$.  Since inverting a permutation of $m$ points requires $\Omega(\sqrt{m})$ quantum queries \cite{Amb02}, $|\Sigma| = \Theta(n^{k-1})$, and $|S| = \Theta(n^k)$, this shows that constructive semigroup membership requires $\Omega(\sqrt{n^{k-1}}) = \Omega(|S|^{\frac{1}{2} - \frac{1}{2k}})$ queries.

To construct the black-box semigroup, we specify an encoding 
\[
  \enc\colon S \to
    \{(a_1,\ldots,a_k) \in \N_0^k\colon 1 \le a_1+\cdots+a_k < n\} 
    \cup \Sigma
    \cup \{\neutral\}
\]
defined by
\begin{align*}
  \enc(g_1^{a_1} \cdots g_k^{a_k}) &\defeq (a_1,\ldots,a_k) 
  && \text{if $a_1+\cdots+a_k < n$} \\
  \enc(g_1^{a_1} \cdots g_{k-1}^{a_{k-1}} g_k^{n - a_1 - \cdots - a_{k-1}})
  &\defeq \pi(a_1,\ldots,a_{k-1}) \\
  \enc(\neutral) &\defeq \neutral.
\end{align*}
We can compute $\enc(g h)$ using at most one call to $\pi$ given the encodings $\enc(g),\enc(h)$ of any $g,h \in S$.  Now suppose we can solve the constructive membership problem for some $\sigma \in \Sigma$ with respect to the generators $g_1,\ldots,g_k$ with encodings $(1,0,\ldots,0),\ldots,(0,\ldots,0,1)$.  Then we can find the values $a_1,\ldots,a_{k-1}$ such that $\enc(g_1^{a_1} \cdots g_{k-1}^{a_{k-1}} g_k^{n-a_1-\cdots-a_{k-1}}) = \sigma$, so that $(a_1,\ldots,a_{k-1}) = \pi^{-1}(\sigma)$, thereby inverting $\pi$.
\end{proof}

In fact, for any fixed $k$, this lower bound is nearly tight.

\begin{theorem}\label{thm:loweropt}
Given a black-box semigroup $S = \<g_1,\ldots,g_k\>$ and an element $x \in S$, there is a quantum algorithm to solve the constructive membership problem for $x$ with respect to $g_1,\ldots,g_k$ in time $|S|^{\frac{1}{2} - \frac{1}{2k}+o(1)}$. Furthermore, there is an algorithm using only $|S|^{\frac{1}{2} - \frac{1}{2k}} \poly(\log |S|)$ quantum queries.
\end{theorem}

To prove this, we use the following simple observation.

\begin{lemma}\label{lem:small}
Let $S$ be a finite abelian semigroup and let  $x,g_1,\ldots,g_k \in S$. Let $(a_1,\ldots,a_k)$ be the lexicographically first $k$-tuple from $\N_0^k$ such that $x=g_1^{a_1} \cdots g_k^{a_k}$. Then $(a_1+1)\cdots(a_k+1)\leq |S|$.
\end{lemma}

\begin{proof}
Assume for a contradiction that $(a_1+1)\cdots(a_k+1) > |S|$. Then, by the pigeonhole principle, there must exist $c_1,\ldots,c_k,d_1,\ldots,d_k \in \N_0$ with $c_i,d_i\leq a_i$ (for all $i=1,\ldots,k$) such that $g_1^{c_1} \cdots g_k^{c_k} = g_1^{d_1} \cdots g_k^{d_k}$ and $(c_1,\ldots,c_k) \ne (d_1,\ldots,d_k)$. Suppose without loss of generality that $(c_1,\ldots,c_k)$ is lexicographically smaller than $(d_1,\ldots,d_k)$. Let $b_i \defeq a_i + c_i - d_i$ for all $i$, and note that $a_i - d_i \geq 0$. Thus $g_1^{a_1} \cdots g_k^{a_k} = g_1^{d_1} \cdots g_k^{d_k} g_1^{a_1 - d_1} \cdots g_k^{a_k - d_k}$ and $g_1^{b_1} \cdots g_k^{b_k} = g_1^{c_1} \cdots g_k^{c_k} g_1^{a_1 - d_1} \cdots g_k^{a_k - d_k}$. This implies $g_1^{b_1} \cdots g_k^{b_k} = x$. Also, for the first index $i$ with $c_i\neq d_i$, we have $c_i<d_i$. Therefore $(b_1,\ldots,b_k)$ is lexicographically smaller than $(a_1,\ldots,a_k)$, a contradiction.
\end{proof}

\begin{proof}[Proof of \thm{loweropt}]
By \lem{small}, there is some $j \in \{1,\ldots,k\}$ such that $\prod_{i \ne j} (a_i+1) \le |S|^{(k-1)/k}$.  To see this, note that $\prod_{j=1}^k \prod_{i \ne j} (a_i+1) = \big(\prod_{j=1}^k (a_j+1)\big)^{k-1} \le |S|^{k-1}$. Thus, for each $j \in \{1,\ldots,k\}$, we perform a Grover search \cite{Gro97} over $(a_1,\ldots,a_{j-1},a_{j+1},\ldots,a_k)\in \N^{k-1}$ with $\prod_{i \ne j} (a_i+1) \le |S|^{(k-1)/k}$, where for each $(k-1)$-tuple we use \lem{sdlog} (with $y=\prod_{i \ne j} g_i^{a_i}$ and $g=g_j$) to find $a_j$ such that $x=g_1^{a_1} \cdots g_k^{a_k}$ (or to exclude its existence). The running time of this procedure is $k |S|^{\frac{k-1}{2k}+o(1)} = |S|^{\frac{1}{2} - \frac{1}{2k}+o(1)}$. Using the query-efficient (but not time-efficient) algorithm for the dihedral hidden subgroup problem in place of Kuperberg's algorithm, we require only $|S|^{\frac{1}{2} - \frac{1}{2k}} \poly(\log |S|)$ queries.
\end{proof}

While \thm{lower} shows that the constructive membership problem is provably hard in black-box semigroups, the problem is also known to be NP-hard in explicit semigroups.  In particular, Beaudry proved NP-completeness of membership testing in abelian semigroups of transformations of (small) finite sets \cite{Bea88}.

\section{Discussion}

We have considered quantum algorithms for the semigroup discrete log problem and some natural generalizations thereof.  While discrete logs can be computed efficiently by a quantum computer even in semigroups, the shifted semigroup discrete log problem appears comparable in difficulty to the dihedral hidden subgroup problem, and the constructive membership problem in a black-box semigroup with respect to multiple generators is provably hard.  Thus, while hardness of the discrete log problem in semigroups is not a good assumption for quantum-resistant cryptography, one might build quantum-resistant cryptosystems based on the presumed hardness of other problems in semigroups.

Testing membership in abelian semigroups is related to a cryptographic problem known as the semigroup action problem (SAP) \cite{MMR07}. Given an (abelian) semigroup $S$ acting on a set $M$ and two elements $x,y\in M$, the SAP asks one to find an element $s\in S$ such that $x=sy$. Constructive membership testing in a monoid (i.e., a semigroup with an identity element, which can be adjoined artificially if necessary) is an instance of SAP: consider $S$ acting on itself by multiplication and let $y$ be the identity. (More precisely, to obtain a decomposition with respect to generators $g_1,\ldots,g_k$, consider the natural action of $\<g_1\> \times \cdots \times \<g_k\>$ on $S$.) On the other hand, the SAP over an abelian semigroup can be reduced to membership of $x$ in a subsemigroup generated by $y$ and $S$ of the abelian semigroup $S'=S\cup M\cup \{0\}$ with a semigroup operation that naturally extends the multiplication of $S$ and the action of $S$ on $M$. In particular, the SAP for a cyclic semigroup action reduces to an instance of the shifted discrete log problem discussed in \sec{sdlog}.

A natural open question raised by our work is the quantum complexity of the shifted semigroup discrete log problem: is this task indeed as hard as the DHSP, or is there a faster algorithm using additional structure?  In general, it might also be interesting to develop new quantum-resistant cryptographic primitives based on hard semigroup problems.

\section*{Acknowledgments}

We thank Rainer Steinwandt for suggesting the problem of computing discrete logarithms in semigroups and for helpful references.
We thank Robin Kothari for pointing out that the lower bound of \thm{lower} generalizes from $k=2$ to $k>2$.
We also thank the Dagstuhl research center and the organizers of its 2013 seminar on Quantum Cryptanalysis, where this work was started.

AMC received support from NSERC, the Ontario Ministry of Research and Innovation, and the US ARO/DTO.
GI received support from the Hungarian Research Fund (OTKA) and from the Centre for Quantum Technologies at the National University of Singapore.



\begin{thebibliography}{10}

\bibitem{Amb02}
Andris Ambainis, \emph{Quantum lower bounds by quantum arguments}, Journal of
  Computer and System Sciences \textbf{64} (2002), no.~4, 750--767,
  \mbox{\href{http://arxiv.org/abs/quant-ph/0002066}{arXiv:quant-ph/0002066}},
  preliminary version in STOC 2000.

\bibitem{BT13}
Matan Banin and Boaz Tsaban, \emph{A reduction of semigroup {DLP} to classic
  {DLP}}, \mbox{\href{http://arxiv.org/abs/1310.7903}{arXiv:1310.7903}}.

\bibitem{Bea88}
Martin Beaudry, \emph{Membership testing in commutative transformation
  semigroups}, Information and Computation \textbf{79} (1988), no.~1, 84--93,
  preliminary version in ICALP 1987.

\bibitem{CD10}
Andrew~M. Childs and Wim van Dam, \emph{Quantum algorithms for algebraic
  problems}, Reviews of Modern Physics \textbf{82} (2010), no.~1, 1--52,
  \mbox{\href{http://arxiv.org/abs/0812.0380}{arXiv:0812.0380}}.

\bibitem{EH00}
Mark Ettinger and Peter H{\o}yer, \emph{On quantum algorithms for
  noncommutative hidden subgroups}, Advances in Applied Mathematics \textbf{25}
  (2000), 239--251,
  \mbox{\href{http://arxiv.org/abs/quant-ph/9807029}{arXiv:quant-ph/9807029}}.

\bibitem{FIMSS03}
Katalin Friedl, G{\'a}bor Ivanyos, Fr{\'e}d{\'e}ric Magniez, Miklos Santha, and
  Pranab Sen, \emph{Hidden translation and translating coset in quantum
  computing}, to appear in SIAM Journal on Computing,
  \mbox{\href{http://arxiv.org/abs/quant-ph/0211091}{arXiv:quant-ph/0211091}},
  preliminary version in STOC 2003.

\bibitem{Gro97}
Lov~K. Grover, \emph{Quantum mechanics helps in searching for a needle in a
  haystack}, Physical Review Letters \textbf{79} (1997), no.~2, 325--328,
  \mbox{\href{http://arxiv.org/abs/quant-ph/9706033}{arXiv:quant-ph/9706033}},
  preliminary version in STOC 1996.

\bibitem{How95}
John~M. Howie, \emph{Fundamentals of semigroup theory}, LMS Monographs,
  vol.~12, Oxford University Press, 1995.

\bibitem{IMS03}
G{\'a}bor Ivanyos, Fr{\'e}d{\'e}ric Magniez, and Miklos Santha, \emph{Efficient
  quantum algorithms for some instances of the non-abelian hidden subgroup
  problem}, International Journal of Foundations of Computer Science
  \textbf{14} (2003), no.~5, 723--739,
  \mbox{\href{http://arxiv.org/abs/quant-ph/0102014}{arXiv:quant-ph/0102014}},
  preliminary version in SPAA 2001.

\bibitem{KKS13}
Delaram Kahrobaei, Charalambos Koupparis, and Vladimir Shpilrain, \emph{Public
  key exchange using matrices over group rings}, Groups Complexity Cryptology
  \textbf{5} (2013), no.~1, 97--115,
  \mbox{\href{http://arxiv.org/abs/1302.1625}{arXiv:1302.1625}}.

\bibitem{Kup05}
Greg Kuperberg, \emph{A subexponential-time quantum algorithm for the dihedral
  hidden subgroup problem}, SIAM Journal on Computing \textbf{35} (2005),
  no.~1, 170--188,
  \mbox{\href{http://arxiv.org/abs/quant-ph/0302112}{arXiv:quant-ph/0302112}}.

\bibitem{MMR07}
G\'erard Maze, Chris Monico, and Joachim Rosenthal, \emph{Public key
  cryptography based on semigroup actions}, Advances in Mathematics of
  Communications \textbf{1} (2007), 489--507,
  \mbox{\href{http://arxiv.org/abs/cs/0501017}{arXiv:cs/0501017}},
  preliminary version in ISIT 2002.

\bibitem{MW97}
Alfred~J. Menezes and Yi-Hong Wu, \emph{The discrete logarithm problem in
  {$GL(n, q)$}}, Ars Combinatoria \textbf{47} (1997), 23--32.

\bibitem{ME99}
Michele Mosca and Artur Ekert, \emph{The hidden subgroup problem and eigenvalue
  estimation on a quantum computer}, Proceedings of the 1st NASA International
  Conference on Quantum Computing and Quantum Communication, Lecture Notes in
  Computer Science, vol. 1509, Springer-Verlag, 1999,
  \mbox{\href{http://arxiv.org/abs/quant-ph/9903071}{arXiv:quant-ph/9903071}}.

\bibitem{MU12}
Alexei~D. Myasnikov and Alexander Ushakov, \emph{Quantum algorithm for the
  discrete logarithm problem for matrices over finite group rings}, Cryptology
  ePrint Archive, Report 2012/574.

\bibitem{Reg04}
Oded Regev, \emph{Quantum computation and lattice problems}, SIAM Journal on
  Computing \textbf{33} (2004), no.~3, 738--760,
  \mbox{\href{http://arxiv.org/abs/cs.DS/0304005}{arXiv:cs.DS/0304005}},
  preliminary version in FOCS 2002.

\bibitem{Sch86}
Alexander Schrijver, \emph{Theory of linear and integer programming},
  Wiley-Interscience, 1986.

\bibitem{Sho97}
Peter~W. Shor, \emph{Polynomial-time algorithms for prime factorization and
  discrete logarithms on a quantum computer}, SIAM Journal on Computing
  \textbf{26} (1997), no.~5, 1484--1509,
  \mbox{\href{http://arxiv.org/abs/quant-ph/9508027}{arXiv:quant-ph/9508027}},
  preliminary version in FOCS 1994.

\end{thebibliography}
\end{document}